\documentclass[journal]{IEEEtran}
\usepackage{cite}

\usepackage{graphicx}
\usepackage{booktabs}
\usepackage{color}

\usepackage{amsmath}
\usepackage{amsthm}
\newtheorem{theorem}{\textit{Theorem}}
\newtheorem{definition}{\textit{Definition}}
\newtheorem{property}{\textit{Property}}
\newtheorem{corollary}{\textit{Corollary}}
\newtheorem{remark}{\textit{Remark}}

\usepackage{algorithm}
\usepackage{algorithmic}

\usepackage{array}

\ifCLASSOPTIONcompsoc
 \usepackage[caption=false,font=normalsize,labelfont=sf,textfont=sf]{subfig}
\else
 \usepackage[caption=false,font=footnotesize]{subfig}
\fi
\usepackage{fixltx2e}
\usepackage{dblfloatfix}

\ifCLASSOPTIONcaptionsoff
 \usepackage[nomarkers]{endfloat}
\let\MYoriglatexcaption\caption
\renewcommand{\caption}[2][\relax]{\MYoriglatexcaption[#2]{#2}}
\fi
\usepackage{url}

\begin{document}
%
\title{Electric Vehicle Charging Right Trading: Concept, Mechanism, and Methodology}
%
%
%

\author{Ruike Lyu,~\IEEEmembership{Student Member,~IEEE}, Yuxuan Gu,~\IEEEmembership{Student Member,~IEEE}, and Qixin Chen,~\IEEEmembership{Senior Member,~IEEE}} 
\maketitle

\begin{abstract}
  With the increasing penetration of electric vehicles (EVs),  uncoordinated EV charging and the resulting chaos, disorder, and long waiting times at EV charging stations (EVCSs) will no longer be tolerable. An EV charging right (CR) is the right to reserve a predefined charging service. By purchasing CRs, EVs can reduce their charging waiting time, and the price of CRs can guide EVs toward optimized charging behaviors. In this article, we define CR, propose the CR trading mechanism (CRM), and analyze the effect of CRM on reducing waiting times and mitigating congestion in EV charging. In the proposed CRM, EVs can purchase CRs in advance, and the CRs are used to estimate the waiting time and update the price of charging. Queue theory is utilized in the waiting time estimation, in which the impact of disclosing queue states at EVCSs is considered for the first time. The simulation results verify the accuracy of the waiting time estimation and the effect of the proposed mechanism.
\end{abstract}


%
\IEEEpeerreviewmaketitle
\section*{Nomenclature}
\addcontentsline{toc}{section}{Nomenclature}
\subsection*{Abbreviations}

\begin{IEEEdescription}[\IEEEusemathlabelsep\IEEEsetlabelwidth{$k/K/\mathcal{K}$}]
  
\item[EV] Electric vehicle.
\item[EVCS] Electric vehicle charging station.
\item[PFC] Public fast charging. 
\item[CNO] Charging network operator.
\item[CR] Charging right.
\item[CRP] The charging right platform.
\item[CRM] The charging right trading mechanism.
\item[NENP] The charging management scheme with no waiting time estimation nor dynamic pricing.
\item[TENP] The charging management scheme with waiting time estimation but no dynamic pricing.
\item[SD] Scheduled delay.
\item[OS] The original charging system.
\item[IS] The imaginary charging system.
\item[PDN] Power distribution network.
\item[UTN] Urban transportation network.

\end{IEEEdescription}

\subsection*{Notations}
\begin{IEEEdescription}[\IEEEusemathlabelsep\IEEEsetlabelwidth{$k/K/\mathcal{K}$}]
  
  \item[$i/I/\mathcal{I}$] Index/number/set of EVs.
  \item[$j/J/\mathcal{J}$] Index/number/set of timeframes.
  \item[$k/K/\mathcal{K}$] Index/number/set of EVCSs.
  \item[$j^{0}$] Preferred timeframe for an EV to get charged.
  \item[$j^{*}$] Optimal timeframe for an EV to get charged.
  \item[$k^{*}$] Optimal EVCS for an EV to get charged.
  \item[$crid$] Identification of a CR.
  \item[$usrid$] Identification of an EV.
  \item[$info$] Relevant information of a CR.
  \item[$e$] Energy amount of a CR.
  \item[$e^{\rm{char}}$] Already charged energy amount of a CR.
  \item[$p(j)$] Unit price of charging in timeframe $j$.
  \item[$w(j)$] Estimated waiting time in timeframe $j$.
  \item[$U(\cdot)$] Loss function of an EV user.
  \item[$\beta^{\rm{TT}}$] User coefficient of total time.
  \item[$\beta^{\rm{SDE}}$] User coefficient of SD early.
  \item[$\beta^{\rm{SDL}}$] User coefficient of SD late.
  \item[$l^{\rm{o}}$] Trip origin of an EV.
  \item[$l^{\rm{d}}$] Trip destination of an EV.
  \item[$l^{\rm{c}}(k)$] Location of EVCS $k$.
  \item[$t^{\rm{curr}}$] Current time.
  \item[$t^{\rm{d}}$] Departure time of an EV.
  \item[$t^{\rm{f}}$] Time for an EV to finish charging.
  \item[$t^{\rm{a}}(k)$] Arrival time of an EV at EVCS $k$.
  \item[$t^{\rm{oc}}(k)$] Traveling time from the trip origin of an EV to EVCS $k$.
  \item[$t^{\rm{cd}}(k)$] Traveling time from EVCS $k$ to the trip destination of an EV.
  \item[$t^{\rm{avail}}(k)$] Time for an available charging pile at EVCS $k$.
  \item[$t^{\rm{w}}(k)$] Waiting time of an EV at EVCS $k$.
  \item[$c$] Number of charging piles.
  \item[$r$] Operating rate of a charging pile.
  \item[$\rm{queue}$] Queue state of an EVCS.
  \item[$\rho$] Service intensity of the charging system.
  \item[$\lambda$] Arrival rate of the charging system.
  \item[$\tau$] Mean charging time of the charging system.
  \item[$\sigma^2_{\rm{a}}$] Relative squared coefficient of variation for interarrival time.
  \item[$\sigma^2_{\rm{c}}$] Relative squared coefficient of variation for charging time.
  \item[$W$] Total time cost of the EV users in a timeframe.
  
  \item[$\Delta t$] Duration of a timeframe.
  \item[$E$] Mean number of EVs generated in a timeframe.
  \item[$\delta z$] A small change of variable $z$.
  \item[$|Z|$] Cardinal of set $Z$.

  \end{IEEEdescription}

\section{Introduction}
%
%
%
%
\IEEEPARstart{T}{he} number of electric vehicles (EVs) is increasing substantially worldwide, resulting a growing demand for EV charging. The most important locations for EV charging are home, work, and then public locations \cite{hardman_review_2018}, in which public charging, although not the most preferred currently, is receiving increasing popularity \cite{morrissey_future_2016}. In China, approximately 80\% of EVs will be charged at public EV charging stations (EVCSs) by 2030, where fast charging will play an important role \cite{mckinsey-report}. However, typical EV charging time, even if for rapid charging, is much longer than refueling gasoline vehicles. For example, operating at 90 kW, it takes 33 minutes for the NIO ec6 model\cite{ec6} to charge the 100 kWh battery pack to half capacity. In addition, uncoordinated charging and the consequent long waiting time could lead to even more loss of driving comfort and social welfare \cite{IARS}. It can be foreseen that limited charging piles and dramatically increasing EVs will aggravate the problem of charging waiting time, bringing challenges for charging management.

For the smart grid, EV charging impacts on load capacity, power quality, economy, and environment are major concerns \cite{shaukat_survey_2018}. In recent studies, minimizing energy costs while avoiding transformer overloads \cite{silva_coordination_2020}, the day-ahead optimal reserve management problem \cite{liu_optimal_2021}, and charging management providing demand response \cite{sadeghianpourhamami_definition_2020, li_constrained_2020} remain topics of interest. However, from EV users' point of view, range anxiety, long charging times, and inconvenient and insufficient charging infrastructure are more concerning\cite{shen_optimization_2019}. This is true especially for public fast charging (PFC), which is more closely related to EV users’ daily activities, as the time spent charging will affect the subsequent activities \cite{utility-2007, utility-2017}. Unfortunately, PFC EVCSs are also where congestions usually take place, which could mean hours of waiting time. Existing studies concerning the waiting time in EV charging mainly focus on EVCS selection and EV routing, in the context of PFC \cite{elghitani_efficient_2021,moradipari_pricing_2020,tucker_online_2020, MWT,trip-duration-2018,subscribe-2017,reservation-2020,Wei-2018, Cui-2021, enhance-2020, MS-2021,charging-navigation}.

Reference \cite{MWT} first addressed the driving comfort issue in EV charging scheduling,
and the basic solution is to select and reserve the EVCSs with minimum
waiting time for EVs. Reference \cite{trip-duration-2018} introduced
traffic conditions into EVCS selection to minimize the total time
(including traveling time on the road and waiting time at the EVCSs)
of EV users. 
In \cite{MWT,trip-duration-2018}, the real-time
locations, speeds, arrival times, and energy demands of EVs are
assumed to be monitored and available for waiting time estimation and
EVCS selection. Similar assumptions and methodology can be seen in
\cite{subscribe-2017,reservation-2020}. We call this modeling
framework the real-time model. The simulation-based study in
\cite{subscribe-2017} suggests that when EVs select EVCSs based on the
real-time queue states (i.e., the time for the EVs at the EVCS to
finish charging) and make charging reservations accordingly, the queue
lengths of the EVCSs are more balanced, meaningfully reducing the
waiting time. However, the real-time model has difficulty in
estimating waiting times several hours in advance, as the queue states
are not available then. Furthermore, the analytical relationship 
between the average waiting time and other parameters, such as the
arrival rate of EVs, has not been captured in the real-time model, but
the relationship is needed for charging pricing \cite{pricing-2017}
and when considering the interdependence between the urban transportation network (UTN) and power distribution network (PDN) \cite{Wei-2017}.

A different modeling of EVCS waiting time is presented in
\cite{Wei-2018, MS-2021,charging-navigation, Cui-2021,enhance-2020},
where the steady-state distribution of traffic flows, instead of the
real-time states of EVs and EVCSs, is used to estimate the average
waiting time (we call the similar modeling the traffic flow model). In
the so-called traffic flow model, the average waiting time in an EVCS
is approximated by the Davidson function developed in
\cite{davidson1966flow} (as in \cite{Wei-2018, Cui-2021, enhance-2020,
  MS-2021}), or by the M/M/n queueing model (as in
\cite{charging-navigation}), formulated as a function of the EV flow
and the capacity of the EVCS. Utilizing the user equilibrium (UE)
model \cite{traffic-flow-1}, the estimated average waiting time of the
EVCSs can be leveraged to determine the distribution of the traffic
flow and enhance the coordinated operation of the PDN and UTN
\cite{enhance-2020}. In \cite{Wei-2018}, the electricity prices, route
selections, time cost and charging opportunities are encapsulated in a
convex traffic assignment problem. With the aim of reducing time and
electricity costs, Reference \cite{Cui-2021} models the optimal charging
pricing problem of EVCSs, where the charging demand is assumed to have
already been forecast, and the steady-state in a single period is considered. In \cite{MS-2021}, while determining the optimal charging price of the EVCSs, the charging demand is analyzed using a modified UE traffic assignment problem with elastic traveling demand and different charging prices. However, all these methods apply only to the steady state, and they rely on accurate EV charging forecast for the purpose of price-setting and decision-making \cite{al-ogaili_review_2019}. Furthermore, even though waiting time can be meaningfully reduced utilizing the real-time queue states~\cite{MWT,trip-duration-2018,subscribe-2017,reservation-2020}, the real-time queue states are not leveraged in the traffic flow model.

\begin{table}[!t]
  \renewcommand{\arraystretch}{1.0}
  \caption{Comparison of Relevant Literature}
  \label{tab-ref}
  \centering
  \setlength{\tabcolsep}{1mm}{
  \resizebox{1.0\columnwidth}{!}{
  \begin{tabular}{cccccc}
    \toprule
    Ref. & \begin{tabular}[c]{@{}c@{}}Waiting time\\ model\end{tabular} & \begin{tabular}[c]{@{}c@{}}Analytically\\ expressed\end{tabular} & \begin{tabular}[c]{@{}c@{}}Utilizing Real-time\\ queue state\end{tabular} & \begin{tabular}[c]{@{}c@{}}Available\\ in advance\end{tabular} & \begin{tabular}[c]{@{}c@{}}Forecast-\\ independent\end{tabular} \\ 
    \midrule
    \begin{tabular}[c]{@{}c@{}} \cite{MWT,trip-duration-2018,subscribe-2017,reservation-2020} \end{tabular}& \begin{tabular}[c]{@{}c@{}}Real-time\\ model\end{tabular} &  $\times$ & $ \surd $  &  $\times$ & $ \surd $  \\ 
    \midrule
    \begin{tabular}[c]{@{}c@{}} \cite{Wei-2018, MS-2021,charging-navigation, Cui-2021,enhance-2020} \end{tabular}& \begin{tabular}[c]{@{}c@{}}Traffic\\ flow model\end{tabular} & $ \surd $  &  $\times$ & $ \surd $  &  $\times$ \\ 
    \midrule
    \begin{tabular}[c]{@{}c@{}}CRM in timeframe\\ selection\end{tabular} & \begin{tabular}[c]{@{}c@{}}Traffic\\ flow model\end{tabular} & $ \surd $  & $ \surd $  & $ \surd $  & $ \surd $  \\ 
    \midrule
    \begin{tabular}[c]{@{}c@{}}CRM in EVCS\\ selection\end{tabular} & \begin{tabular}[c]{@{}c@{}}Real-time\\ model\end{tabular} &  $\times$ & $ \surd $  &  $\times$ & $ \surd $  \\ 
  \bottomrule
  
  \end{tabular}
  }}
\end{table}

Utilizing real-time information in EV charging management is an interesting trend \cite{mukherjee_review_2015}. At present, there are a few web and mobile-based applications (APPs)
for EV charging management, such as E Charge \cite{e-charge}, which
displays current number of EVs in each EVCS. With the application of more
reliable and privacy-preserving techniques, such as blockchains
\cite{privacy-preserving,2020Coordinating-blockchain}, the real-time
queue states of EVCSs may be available to EVs, making the assumptions and the
select-reserve method of the real-time model realistic. By reserving a
charging pile in the selected EVCS, an EV can avoid a long waiting
time. However, other EVs may suffer from longer waiting times due to
the reduced availability of charging piles after the reservation. We
argue that the waiting time should be modeled considering the
utilization of EVCS real-time queue states so that the external cost of the
reservation can be measured. Moreover, when viewing the reservation as a right, pricing can be introduced to allocate the rights to make reservations, especially when charging piles are scarce.

In this article, we define EV charging right (CR) and propose the CR trading mechanism (CRM) for EVCSs and EVs to participate in. First, the concept and mechanism design of CR trading is clarified: EVs can select the timeframe to get charged and buy CRs accordingly, which can be used to reserve the selected EVCSs. Then, we develop the framework for waiting time estimation and CR pricing to reflect the external cost of CRs. Finally, simulations are conducted to verify the accuracy of the proposed method for waiting time estimation and the effect of the proposed mechanism on time cost reduction.

A comparison of EV charging waiting time modeling in the literature is summarized in Table~\ref{tab-ref}. Different from existing research, we take into account the impact of utilizing EVCS queue states on the average waiting time of a charging system. We theoretically analyze the lower bound of the average waiting time, which, as we prove, can be approached under certain conditions. The average waiting time is then analytically expressed, formulated as a function of several system parameters, which can be obtained from the CRs, without relying on forecasts or specific user models. Moreover, the external cost of charging is integrated into the price of charging, which is based on the increase of average waiting time of the charging system.

Our contributions are as follows:
 
\begin{itemize}
  \item We model and analytically express the average waiting time in an EV charging system where the real-time queue states of the EVCSs are utilized by EVs to select the EVCS for charging.
  
  \item We introduce the concept of CR and the corresponding CRM, a charging management scheme aimed at reducing waiting time, where the forecast-independent waiting time estimation and price of charging are provided to EVs for decentralized charging planning.
  
  \item Our simulation in the case study shows the accuracy of the proposed methodology for waiting time estimation and the effect of the CRM on waiting time reduction.
\end{itemize} 

The remainder of this paper is organized as follows: Section~\ref{sec_mechanism} presents the concept of CRs and the corresponding mechanism for CR trading. Section~\ref{sec_model} presents the modeling and theoretical analysis of waiting time in EV charging. Section~\ref{sec_methodology} presents the framework for estimating waiting time and CR pricing in advance. Section~\ref{sec_case} presents the case study. Section~\ref{sec_conclusion} draws conclusions and prospects for future work.

\section{Design of Charging Right Trading Mechanism}\label{sec_mechanism}
In this section, the concept of CRs and the design of the CR trading mechanism (CRM) are introduced. We consider an urban PFC system, where EVCSs are deployed in the traffic network and EVs travel from their origins to destinations, being charged during the trip. All the EVCSs and EVs are accessible on the Internet, so that the estimated waiting time and charging price are available to them, and EVs can send requests to buy CRs or make reservations on the charging right platform (CRP).

\subsection{Consideration of entities in the CRM}
\paragraph{\textbf{Electric vehicle}}
EVs with charging demand are considered. The total charging cost of an
EV user includes a monetary cost and a time cost. Based on the price and waiting time queried from the CRP, EV users decide when and where to get charged and buy the corresponding CR in advance. The distance between the origins of EVs and the EVCSs is negligible compared with the maximum range of EVs (600 km for ec6 \cite{ec6}), and the energy consumption on the road can be neglected, as assumed in \cite{Wei-2018,Cui-2021,MS-2021}. We also assume that each EV has its own agent, such as mobile-based APPs, so the timeframe and the EVCS for charging are rationally selected. Hereafter, we use EV user and EV interchangeably.

\paragraph{\textbf{Electric vehicle charging station}}
The main types of EVCSs include EVCSs for PFC, battery swapping, and wireless charging \cite{ding_technical_2020}. In the CRM, we consider PFC EVCSs for commercial purposes, where EVs pay for charging according to the energy amount and charging time. Although the EVCSs are possibly owned by different companies, we assume that they are managed by a common charging network operator (CNO) as in \cite{Cui-2021, moradipari_pricing_2020}, instead of setting prices independently to compete with each other as in \cite{MS-2021}. Moreover, the EVCS queue states are monitored and available to EV users, as assumed in \cite{MWT,trip-duration-2018,subscribe-2017,reservation-2020}. The assumptions are based on the observation that increasing numbers of EVCSs have been accessed in standardized charging management platforms such as E Charge \cite{e-charge}, rather than operating on their own. 

Charging systems for domestic garages, apartment complexes, and commercial complexes are dissimilar to PFC systems in business models and charging modes \cite{rahman_review_2016}. In addition, EVCSs can be aggregated to participate in the electric power market and integrate renewable energy \cite{teng_technical_2020}. We leave the consideration of these different charging modes to future work.

\subsection{Concept and circulation of charging rights}

\begin{figure}[!t]
  \centering
    \includegraphics[width=3.0in]{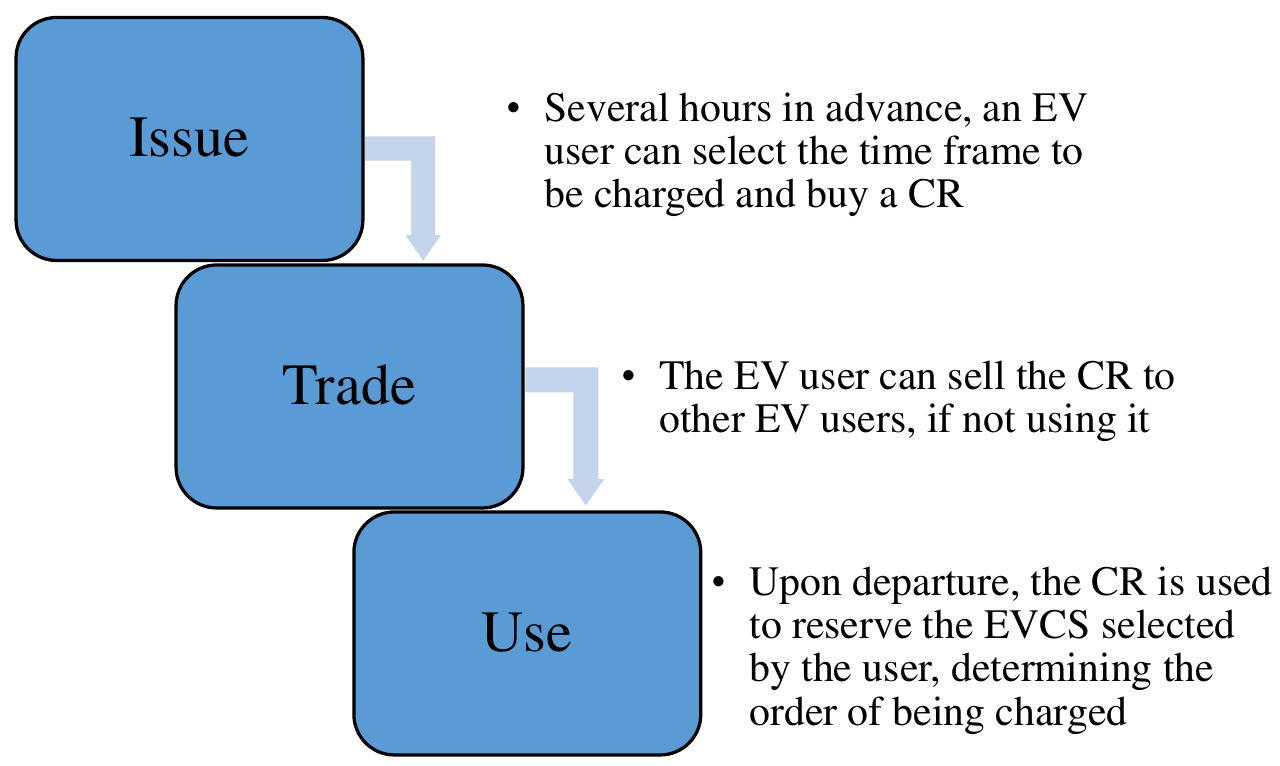}
  \caption{Circulation of a charging right.}
  \label{chap1-f2}
\end{figure}

A CR is the right to make a reservation and charge an EV with a given
amount of energy within a stipulated timeframe. The CR can be denoted
as ${\rm CR} = (crid, usrid, j, e, info)$, where $crid$ and $usrid$
denote the identification (id) of the CR and its owner, respectively,
$j$ is the timeframe in which the CR can be used to make a
reservation, $e$ is the amount of energy to be charged, and $info$
refers to other relevant information. When the CR is used, the
selected EVCS $k$ and the arrival time $t^{\rm{a}}$ at the EVCS are recorded as $info = (k, t^{\rm{a}})$. For example, the CR of $(cr1, usr1, 8:00-10:00, 30 \ {\rm kWh}, info1)$ is registered with an id of $cr1$ and can be used by $usr1$ to reserve an EVCS and get charged during any time between 8:00 and 10:00.

Fig.~\ref{chap1-f2} describes the circulation of a CR, which can be divided into three stages: issue, trade, and use. First, requested by an EV, the CRP issues the CR and sets it owned by the EV at the current public price (as presented in \ref{pricing}). Second, the CR owner can set a price for the CR and sells it to other EVs willing to buy one. Third, the CR is verified and used to charge an EV at any EVCS selected by the owner.

\begin{figure}[!t]
  \centering
    \includegraphics[width=3.0in]{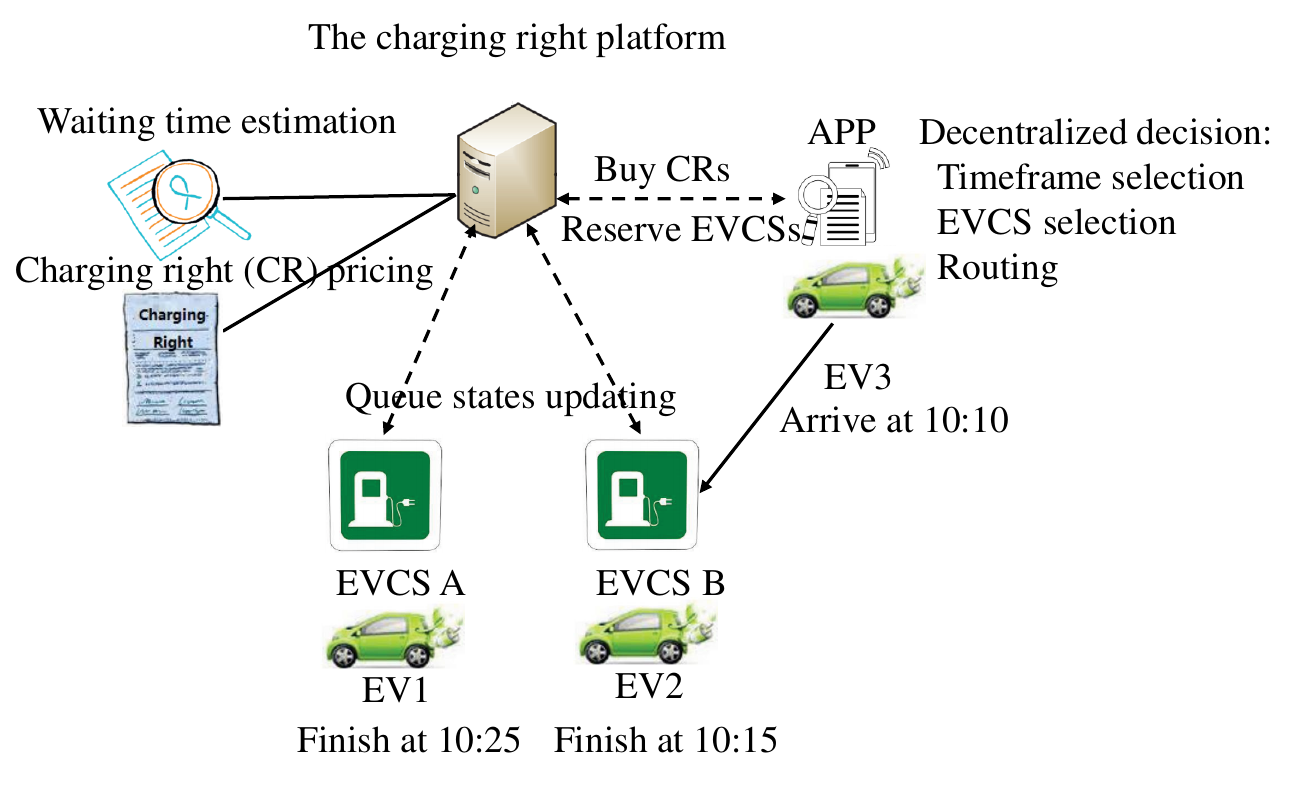}
  \caption{Overview of the CRM.}
  \label{chap1-f1}
\end{figure}

\subsection{Protocol of the CRM}\label{protocol}
The overview of the proposed CRM is shown in Fig.~\ref{chap1-f1}. The detailed protocol is presented as follows. Since we present the prototol in the view of an EV, the index of EVs $i$ is omitted for simplicity.

\paragraph{\textbf{Select the timeframe}}
A few hours in advance, an EV selects the timeframe $j^*$ to get
charged, for $j \in \mathcal{J} = \{1,...,J\}$. The unit charging
prices $p(j)$ and estimated waiting time $w(j)$ in future time
frames are acquired from the CRP. The user inputs the preferred timeframe for charging $j^0$ and the energy demand $e$ to the agent. The
optimal timeframe is then selected locally by the agent, as captured by:
\begin{equation}\label{select-time-frame}
  j^* = \underset{j \in \mathcal{J}}{\rm argmin} \ U(e \cdot p(j), w(j), (j-j^0)),
\end{equation}
which is a tradeoff based on the loss function $U( \cdot )$ that
describes the loss of user utility due to the price $e \cdot p(j)$,
the waiting time $w(j)$ and the scheduled delay (SD) $j-j^0$, based
on the assumption that users have their preferred departure time, and
any SD off that reduces their utility \cite{utility-2017}. For the
sake of discussion, in the case study, we take the linearized form of
$U( \cdot )$ developed in \cite{utility-2007}; however, the CRM is not restricted by specific user models. Let $\beta^{\rm{TT}}$, $\beta^{\rm{SDE}}$ and $\beta^{\rm{SDL}}$ represent the user coefficients of total time, SD early and SD late, respectively, which can be input by the user or set as default values. Thus we have:
\begin{equation}\label{utility-function}
  \begin{aligned}
    U = & e \cdot p(j) + \beta^{\rm{TT}} w(j) \\
        & + \beta^{\rm{SDE}}{\rm max} \{j^0-j, 0\} + \beta^{\rm{SDL}}{\rm max} \{j-j^0, 0\}.
  \end{aligned}
\end{equation}

The timeframe selection when buying CRs from other users differs from
the above, as the timeframe $j$ and amount of energy $e$ of the
existing CRs have already been stipulated, and buyers are likely to be
interested only in those CRs with proper $j$ and $e$ for them. 

\paragraph{\textbf{Buy the CR}}\label{buy-the-cr}
After timeframe selection, the EV sends a request to buy the corresponding CR on the CRP and pays for it. There are two ways of buying CRs, either from the CRP or from other users. 

When buying a CR from the CRP, the request $(usrid, j, e)$ for the selected timeframe $j$ and the amount of energy $e$ is sent to the CRP. ${\rm CR} = (crid, usrid, j, e, info)$ is then issued by the CRP, where $usrid$, $j$ and $e$ are from the request, $crid$ is generated, and $info$ is to be fulfilled later. The price of the CR is set according to the price mechanism presented in \ref{pricing}, and the real-time price and estimated waiting time are updated once the transaction is confirmed.

When buying a CR from other users, which has already been issued on
the CRP as described above, the price is set by the current owner, and the $usrid$ of the traded CR is changed to the new owner after the transaction is confirmed. 

\paragraph{\textbf{Sell the CR} (optional)}
The sale of owned CRs is allowed because in some instances, one may
have bought a CR and then canceled the trip plan, and thus selling the CR can reduce the loss. Moreover, the transfer of CRs can further improve social welfare, possibly when the public price has risen substantially since the CR was bought, thus making it profitable to sell the CR to other users with inelastic demands. The pricing strategy in the CR trading between users is beyond the scope of this article, and it is up to the owners to decide. 

\paragraph{\textbf{Select the EVCS}}
Any time in the timeframe stipulated by the CR, the EV acquires the
real-time queue states of the EVCSs from the CRP and selects an
optimal EVCS $k^*$, for $k \in \mathcal{K} = \{1,...,K\}$. First, the
user inputs the departure time $t^{\rm{d}}$, the trip origin $l^{\rm{o}}$, and
destination $l^{\rm{d}}$. Then, the estimated available time $t^{\rm{avail}}(k)$
for an idle charging pile (Algorithm~\ref{alg_ATE}), the charging rate
$r(k)$, and the locations of the EVCSs $l^{\rm{c}}(k)$ are acquired from
the CRP. The shortest route that links the origin, the EVCS, and the
destination can be formulated, and the corresponding traveling time
$t^{\rm{oc}}(k)$, $t^{\rm{cd}}(k)$ can be calculated. Although the vehicle routing problem is an NP-hard optimization problem \cite{erdelic_survey_2019}, the routing problem in EV charging based on real-time traffic data has been well studied in \cite{routing-1,routing-2}, and we assume that $t^{\rm{oc}}(k)$
and $t^{\rm{cd}}(k)$ are available to EVs (e.g., from general use navigation software). The famous Dijkstra algorithm is used in the case study. Thus, the estimated arrival time at EVCS $k$ is
\begin{equation}
  t^{\rm{a}}(k) = t^{\rm{d}} + t^{\rm{oc}}(k).
\end{equation}
We assume the EVCSs to be standardized such that the charging rates are identical to $r$. The charging time is simplified to $e/r$, which is the same for all EVCSs. Therefore, the charging time is not included in our discussion of waiting time hereafter. Upon arrival (if all the charging piles have been occupied), the EV has to wait until an idle charging pile appears at $t^{\rm{avail}}$, and therefore the waiting time $t^{\rm{w}}(k)$ at EVCS $k$ is 
\begin{equation}\label{waiting_time}
  t^{\rm{w}}(k) = {\rm max} \{t^{\rm{avail}}(k) - t^{\rm{a}}(k), 0\},
\end{equation}
The EVCS selection is captured by a tradeoff between the waiting time
at the EVCS and the traveling time on the road, with the aim of minimizing total time:
\begin{equation}\label{select-cs}
  k^* = \underset{k \in \mathcal{K}}{\rm argmin} \ t^{\rm{w}}(k) + t^{\rm{oc}}(k) + t^{\rm{cd}}(k).
\end{equation}
Note that in our mechanism, EVs select EVCSs in a distributed manner according to their own wishes, rather than the CNO assigning one for them (as in \cite{moradipari_pricing_2020, elghitani_efficient_2021}). For example (Fig.~\ref{chap1-f1}), the arrival time of EV3 and the queue states of EVCS A and B are used to select EVCS B for EV3. It is also possible that the user chooses to depart later until some EVCSs with less waiting time appear, but the CR can only be used within the timeframe. Note that when selecting the timeframe and the EVCS, the inputs from EVs are only used locally and are not revealed to the CRP.

\paragraph{\textbf{Make a reservation}}
For the selected EVCS, the EV sends a reservation request to the
CRP. The request $(crid, k, t^{\rm{a}})$ consists of the identity of the CR
$crid$, the selected EVCS $k$ and the estimated arrival time
$t^{\rm{a}}$. Then, $(k, t^{\rm{a}})$ is recorded in the $info$ of the CR as a reservation, which is used to determine the order of charging, update the queue state of EVCS $k$, and further update the available time $t^{\rm{avail}}(k)$ by Algorithm~\ref{alg_ATE}. 

\paragraph{\textbf{Get charged at the EVCS}}
The CR is verified at the EVCS, and the EV eventually gets
charged. Upon arrival, the EVCS verifies that the current time is
within the timeframe, the user matches the owner, and the reserved
EVCS matches the EVCS, which are stipulated by the CR. This can be
done through cryptography techniques as in
\cite{privacy-preserving}. Moreover, malicious reservations can be
identified by comparing the arrival time in the reservation and the
current time. If valid, the EV is charged the energy $e$ in the CR. If
there is no free charging pile, the EVs have to wait in a queue based
on the First Come First Serve (FCFS) order of reservations. During charging, the queue states of the EVCSs are monitored and updated to the CRP. We assume that parking without charging while there are waiting EVs can be avoided. On parking without charging, readers may refer to \cite{tucker_online_2020}.

Apart from the above protocol, an EV can simply depart and get charged at an EVCS without buying a CR or making a reservation in advance. However, such an EV is unable to learn about the charging price and waiting time for trip planning in advance. 

\section{Waiting Time Modeling}\label{sec_model}

In this section, we model the waiting time in EV charging and
theoretically analyze the average waiting time of the charging
system. In the CRM, waiting time estimation is needed in timeframe
selection and EVCS selection of the EVs (see Table~\ref{tab-ref}). The
two scenarios differ in the amount of time in advance of charging and
the consequent uncertainty of EV behaviors. 
When an EV departs and selects an EVCS, its arrival time at the EVCS
can be estimated with high accuracy by considering the real-time traffic conditions \cite{routing-1}. Moreover, the real-time EVCS queue states can be leveraged. In contrast, EVs may select a timeframe several hours in advance, and the real-time information is not available then. We start from the waiting time at an EVCS in the real-time scenario.

\subsection{Waiting time at an EVCS}\label{sec_ATE} 

\begin{algorithm}[!t]
  \caption{Available time of an EVCS}
  \label{alg_ATE}

  \begin{algorithmic}[1]
    \renewcommand{\algorithmicensure}{\textbf{Input:}} \ENSURE
    Parameters of the EVCS: number of charging piles $c$, charging rate $r$; monitored by the EVCS: already charged energy $e^{\rm{char}}_i$; from the CRs/reservations: the arrival time $t^{\rm{a}}_i$ and the energy to charge $e_i$

    \renewcommand{\algorithmicensure}{\textbf{Output:}} \ENSURE
        The time for an available charging pile $t^{\rm{avail}}$

    \renewcommand{\algorithmicensure}{\textbf{Data structure:}} \ENSURE
      $\{ t^{\rm{f}}_i|i = 1, ... , I\}$, a queue of the time for the EVs at the EVCS to finish charging.
    
    \renewcommand{\algorithmicensure}{\textbf{A. Construct the queue:}} 
    \ENSURE \
    \STATE queue $= null$, $i = 1$
    \FOR{charging EVs}
        \STATE $t^{\rm{f}}_i = t^{\rm{curr}} + (e_i - e^{\rm{char}}_i)/r$. \\
        $\backslash *$ $t^{\rm{curr}}$ is the current time. $*\backslash$
        
        \STATE Add $t^{\rm{f}}_i$ to queue, $i++$. 

        \STATE Sort queue in ascending order of $t^{\rm{f}}_i$.
    \ENDFOR
    \FOR{waiting EVs and EVs with reservations} 
        \STATE $t^{\rm{w}}_i = {\rm max} \{0, t^{\rm{f}}_{i-c} - t^{\rm{a}}_i\}$. \\
        $\backslash *$ The $i^{\rm th}$ EV has to wait until the $(i-c)^{\rm th}$ EV finishes charging. Perform in ascending order of $t^{\rm{a}}_i$, so that any $t^{\rm{f}}_i$ before is already in queue. $*\backslash$

        \STATE $t^{\rm{f}}_i = t^{\rm{a}}_i + t^{\rm{w}}_i + e_i/r$. \\
        $\backslash *$ Finishing charging after waiting. $*\backslash$
        
        \STATE Add $t^{\rm{f}}_i$ to queue, $i++$.
    \ENDFOR

    \renewcommand{\algorithmicensure}{\textbf{B. Calculate the Available Time}}\ENSURE \
    \STATE $I = |\{t^{\rm{f}}_i\}|$
    \IF{$I < c$}
        \STATE \textbf{return} $t^{\rm{avail}} = t^{\rm{curr}}$. \\
        $\backslash *$ There is currently at least one idle pile. $\backslash *$ 
    \ELSE
        \STATE \textbf{return} $t^{\rm{avail}} = t^{\rm{f}}_{I-c+1}$. \\
        $\backslash *$ An idle pile is available when the $(I-c+1)^{\rm th}$ EV finishes charging. $\backslash *$ 
    \ENDIF
  \end{algorithmic}
\end{algorithm}

Consider an arbitrary EVCS (therefore the index of EVCSs $k$ is omitted for simplicity) with $c$ charging piles, and each can operate at the rate of $r$. If all the charging piles have been occupied, EVs need to wait in a queue based on the FCFS order. Viewing the EVCS as a charging system, we provide the following definitions:

\begin{definition}
  The queue state of a charging system is defined as a queue of the times for the EVs to finish charging in the charging system.
\end{definition}

\begin{definition}
  The available time of a charging system is defined as the time for an idle charging pile in the charging system to be available for charging.
\end{definition}

Clearly, the above definitions also apply to charging systems consisting of more than one EVCS. 

The queue state of the EVCS can be denoted as ${\rm queue} = \{ t^{\rm{f}}_i|i = 1, ... , I\}$, where $t^{\rm{f}}_i$ and $I$ are the time for the $i^{\rm th}$ EV (including those with reservations) to finish charging and the length of the queue, respectively. With the already charged energy $e^{\rm{char}}_i$ monitored by the EVCS, the arrival time $t^{\rm{a}}_i$ from the reservations, and the energy $e_i$ from the CRs, the queue state can be constructed (Algorithm~\ref{alg_ATE}.A). Then, Algorithm~\ref{alg_ATE}.B is called to calculate the available time $t^{\rm{avail}}$ of the EVCS. Finally, the estimated waiting time at the EVCS with respect to the arrival time $t^{\rm{a}}$ of an EV is given by Equation (\ref{waiting_time}). 

\begin{remark}
  The arrival time $t^{\rm{a}}_i$ of the reservations is an estimated result, which introduces error into the waiting time estimation in Algorithm~\ref{alg_ATE} and Equation (\ref{waiting_time}). Since the calculation of $t^{\rm{avail}}$ is affected only if $I > 2c$, which rarely happens, the error of Equation (\ref{waiting_time}) mainly lies in the traveling time estimated by navigation software.
\end{remark}
 
\subsection{Average waiting time of the charging system}
Consider a charging system with $K$ EVCSs, where EVCS $k$ is equipped with $c_k$ charging piles operating at the rate of $r$, for $k \in \mathcal{K} = \{ 1,...,K \}$. The arrival time of an EV at the charging system is defined as the time when it arrives at any of the EVCSs (selected by the EV in some way), and the interarrival time is the time between two neighboring arrivals. We call the above charging system the original system (OS). 

\begin{definition}
  The imaginary system (IS) of a charging system is defined as a mapping of the OS, which
  consists of one EVCS equipped with all the charging piles in the OS
  (i.e., its charging piles $c = \underset{k \in \mathcal{K}}{\Sigma}
  c_k$); when an EV arrives at the OS, an EV with the same amount of
  energy to charge also arrives at the IS; and the EVs in the IS are charged according to the FCFS order.
\end{definition}

\begin{property}\label{property_lower_bound}
  An average waiting time lower bound of the OS is given by the average waiting time of the IS.
\end{property}

\begin{proof}
  As the order of EV charging is set as the FCFS order in both the OS
  and IS, the average waiting time is only determined by the EVCS
  selection of the EVs. The difference between the two systems is that
  when an EV arrives in the OS (at the EVCS selected in some way), it
  may have to wait, even if there are idle charging piles in other
  EVCSs in the charging system. When an EV arrives in the IS, however,
  the EV does not need to wait as long as there are available piles in
  the system. Therefore, the average waiting time in the OS increases
  more than that of the IS if there are waiting EVs while idle piles
  exist in the OS. Otherwise, the average waiting time increases the
  same for the two systems. Therefore, the average waiting time of the
  IS is a lower bound of that of the OS.
\end{proof}

\begin{theorem}\label{theorem_equal}
  If the EVs arrive at the EVCS with the earliest available time, the average waiting time of the OS is equal to the IS.
\end{theorem}

\begin{proof}
  See Appendix \ref{theorem_equal_app}.
\end{proof}

\begin{corollary}\label{corollary_equal}
  If the difference in traveling time with respect to the EVCSs is negligible compared to the difference in waiting time, and the EVs arrive at the EVCS with the minimum total time (\ref{select-cs}), the average waiting time of the OS is equal to the IS.
\end{corollary}

\begin{proof}
  See Appendix \ref{corollary_equal_app}.
\end{proof}

Considering the much higher density of EVCSs in future PFC
systems, we assume the premise in Corollary \ref{corollary_equal} to be
satisfied. Furthermore, with real-time queue states available, EVs can select, reserve, and arrive at the EVCS with the minimum total time (\ref{select-cs}), as in \cite{MWT,trip-duration-2018,subscribe-2017,reservation-2020}. Therefore, the average waiting
time of the OS can be estimated by that of the IS. Nevertheless, the EVCS with the earliest available time might not be the one with the minimum total time given by (\ref{select-cs}) (e.g., an EVCS with an idle charging pile in the far distance), which could lead to underestimating the average waiting time of the OS, and we will investigate whether the error is tolerable in the case study.

Without loss of
generality, we consider the distributions of the interarrival time and
charging time of EVs as general (G). Since EVs select EVCSs on their own, we further assume the distribution of the interarrival time to be independent (GI). Thus, the IS is formulated as the
GI/G/c model. The exact analysis of the queueing model can be
difficult, but simple approximations have been developed and
evaluated. We approximate the average waiting time in the GI/G/c model
with the method concerning heavy traffic conditions developed in \cite{kingman1965heavy}. Let $w$ be the average waiting time in the GI/G/c model and $w(\rho, \tau, \sigma^2_{\rm{a}}, \sigma^2_{\rm{c}}, c)$ represent $w$ as a function of the five parameters, where $\rho = \lambda \tau/ c, \lambda$, and $\tau$ denote the service intensity, arrival rate, and mean charging time, $\sigma^2_{\rm{a}}$, and $\sigma^2_{\rm{c}}$ denote the relative squared coefficient of variation (divided by the mean) of interarrival time and charging time, respectively. In steady stable state (i.e., $\rho < 1$), the average waiting time of a GI/G/c system is approximated by
\begin{equation}\label{mean_waiting_time}
  w(\rho, \tau, \sigma^2_{\rm{a}}, \sigma^2_{\rm{c}}, c) = \tau\frac{(\sigma^2_{\rm{a}} + \sigma^2_{\rm{c}})}{2 c} \frac{\rho}{(1 - \rho)}.
\end{equation}

We conclude that given the service intensity $\rho$, the average waiting time is inversely proportional to the number of charging piles $c$ in the whole charging system. Following convention (instead of using the notaions in other parts of this paper), let $J = (\sigma^2_{\rm{a}} + \sigma^2_{\rm{c}})/2c$, and Equation (\ref{mean_waiting_time}) is in the same form as the Davidson function: $w(x) = \tau Jx/(c - x)$, where the EV flow $x = c \rho$ \cite{davidson1966flow}. However, the waiting time estimation by the Davidson function in \cite{Wei-2018, MS-2021, Cui-2021,enhance-2020} is based on the capacity $c$ and service intensity $\rho$ of a single EVCS, in which the effect of the utilization of EVCS real-time queue states on waiting time reduction is not considered. 

Property 1 and Corollary 1 together show that given the service intensity, the EVCS selections described by (\ref{select-cs}) (i.e., EVs select the EVCS with the minimum total time cost) can lead to the minimum average waiting time, which decreases with the total number of charging piles in the charging system. This inspires us to simply make the queue states available and hand over the selection of EVCSs to EVs instead of pricing differently across the EVCSs because extra pricing does not make the users' EVCS selection optimal, if not worse, in terms of average waiting time. Therefore, the charging pricing in the CRM is timeframe-oriented, and we will discuss it in the next section.

\section{Methodology for Waiting Time Estimation\\ and Charging Right Pricing}\label{sec_methodology}

\subsection{Estimating waiting time in advance}\label{WTE}
Although the arrival time, charging time, and EVCS selections of the
EVs become stochastic in advance, the statistics of CRs can be
leveraged to estimate the waiting time of the charging
system by Equation (\ref{mean_waiting_time}). Nevertheless,
Equation (\ref{mean_waiting_time}) is derived from the steady-state
queueing model, while the charging demand  usually varies over the
day. Therefore, divide the day into timeframes meeting the  following
requirements. (a) In one particular timeframe, the parameters $(\rho,
\tau, \sigma^2_{\rm{a}}, \sigma^2_{\rm{c}})$ are steady. (b) The durations of the
timeframes are long enough that the system is mostly in steady state in each timeframe. For the sake of discussion, we consider 12 timeframes with a duration of $\Delta t = 120$ min (e.g., 17:00$-$19:00) and assume the above requirements to be met.

Let ${\rm CR}_i = (crid_i, usrid_i, j_i, e_i, info_i), i \in
\mathcal{I}$ represent an arbitrary CR. In an arbitrary
timeframe $j$, of the five parameters in Equation
(\ref{mean_waiting_time}), $c$ is the total number of charging piles
and known, $\rho$ is given by $\rho = \lambda \tau/ c$, while $
\lambda, \sigma^2_{\rm{a}}, \tau, \sigma^2_{\rm{c}}$ are the statistics of the interarrival times and charging times of EVs, estimated by
\begin{equation}\label{derive_lambda}
  \lambda(j) = \frac{|\{CR_i | j_i = j;i \in \mathcal{I} \}|}{\Delta t},
\end{equation}
\begin{equation}
  \sigma^2_{\rm{a}}(j) = \frac{1}{\lambda(j)} \underset{j_i = j; i \in \mathcal{I}}{\Sigma} \frac{(t^{\rm{a}}_i - t^{\rm{a}}_{i-1} - \frac{1}{\lambda(j)})^2}{(1/\lambda(j))^2},
\end{equation}
\begin{equation}
  \tau(j) = \frac{1}{\lambda(j)} \underset{j_i = j; i \in \mathcal{I}}{\Sigma} \frac{e_i}{r},
\end{equation}
\begin{equation}\label{derive_sigma}
  \sigma^2_{\rm{c}}(j) = \frac{1}{\lambda(j)} \underset{j_i = j; i \in \mathcal{I}}{\Sigma} \frac{( e_i / r - \tau(j))^2}{\tau(j)^2},
\end{equation}
respectively, where $t^{\rm{a}}_{i-1}$ represents the arrival time of the EV
before EV $i$. However, the arrival time of the EVs is not known in
advance. Therefore, $\sigma^2_{\rm{a}}$ can be estimated with the historical
data recorded in the $info$ of the CRs. 

\subsection{Charging right pricing}\label{pricing}
The price of ${\rm CR} = (crid, usrid, j, e, info)$ is designed to
reflect the external time cost induced by charging an EV with energy
$e$ in timeframe $j$. For the time cost, we focus on the loss of EV user utility due to the traveling time on the road, the waiting time at the EVCS, and the SD of the trip plan. Nevertheless, the money paid for CRs should not be distributed to users to cover their time costs, which may encourage charging during peak hours. Instead, we suggest maintaining the CRP or encouraging participation in the CRM with money.

In the CRM, the net price of charging consists of the CR price,
electricity price, and other components. The pricing of electricity is
beyond the scope of this paper. Although it is included in the CRM, we
focus more on the pricing of CRs and use the same notation for CR
price and the net price of charging. A static electricity price is used in the case study. 

Consider an arbitrary timeframe $j$ during which $I = l \lambda(j)$
CRs have been bought. To simplify notation, $j$ is omitted in the representation of parameters in the following paragraphs (e.g., $\lambda$ for $\lambda(j)$). We assume that no CR is abandoned, and its owner will get charged in the stipulated timeframe. The users follow the CRM protocol and select the optimal EVCSs (\ref{select-cs}) $\{k^*_i | i \in \mathcal{I} = \{0,...,I\} \}$. The time cost $W$ of the users in the timeframe is
\begin{equation}
  W =  \beta^{\rm{TT}} \underset{i \in \mathcal{I}}{\Sigma} (t^{\rm{w}}_i(k^*_i) + t_i^{oc}(k^*_i) + t_i^{cd}(k^*_i)),
\end{equation}
where $\beta ^ {TT}$ is the average user coefficient of total time and
needs to be properly set. Now, a new EV $i^\prime$ buys a CR with
energy $\delta e$ and gets charged within the timeframe. When the
charging is performed, the EVCS selections will change to
$\{k^{\prime*}_i| i \in \mathcal{I} \cup \{i^\prime\}\}$ due to the
change in queue states delivered by $i^\prime$. As we have assumed the
difference in traveling time with respect to the EVCSs to be
negligible, the change in traveling time $(t_i^{oc}(k^{\prime*}_i) +
t_i^{cd}(k^{\prime*}_i)) - (t_i^{oc}(k^*_i) + t_i^{cd}(k^*_i))$ is
negligible. Furthermore, collecting the detailed traveling time and
the change is hardly practical due to computational complexity and
privacy issues. Therefore, we only consider the change in waiting
time. Note that the change in the parameters $\tau, \sigma^2_{\rm{a}}$ and
$\sigma^2_{\rm{c}}$ due to a new EV is negligible. Let $w(\lambda)$ represent
the estimated waiting time as a function of $\lambda$, let $\delta
\lambda = \delta e /r \tau l$ represent the normalized change of
$\lambda$ due to EV $i^\prime$, and let the time cost of the users
changes be calculated by
\begin{equation}
  \delta W = \underset{i \in \mathcal{I} \cup \{i^\prime\}}{\Sigma} \beta^{\rm{TT}} w(\lambda + \delta \lambda) - 
  \underset{i \in \mathcal{I} }{\Sigma} \beta^{\rm{TT}} w(\lambda),
\end{equation}
where $\beta^{\rm{TT}} w(\lambda + \delta \lambda)$ is the estimated waiting time for EV $i^\prime$, and $\underset{i \in \mathcal{I} }{\Sigma} \beta^{\rm{TT}}_i (w(\lambda + \delta \lambda) - w(\lambda))$ is the induced external time cost. As $\delta \lambda \ll \lambda$, $w(\lambda + \delta \lambda)$ is approximated by its first-order Taylor expansion, i.e., $w(\lambda + \delta \lambda) \approx w(\lambda) + (\partial w/\partial \lambda) \delta \lambda$. Thus, the induced external time cost (still represented by $\delta W$) is
\begin{equation}\label{unit_cost}
  \delta W = \underset{i \in \mathcal{I}}{\Sigma} \beta^{\rm{TT}} (  \frac{\partial w}{\partial \lambda})\delta \lambda,
\end{equation}
where $\underset{i \in \mathcal{I}}{\Sigma} \beta^{\rm{TT}} = I \beta^{\rm{TT}}$
is given by $l \lambda \beta^{\rm{TT}}$. Let $P(\delta e, \rho) = \delta W$
represent the price of a CR as a function of $\delta e$ and $\rho$, and substitute Equation (\ref{mean_waiting_time}) into Equation (\ref{unit_cost}) ($\lambda = \rho c/\tau$):
\begin{equation}\label{sale_price}
  P(\delta e, \rho) = \beta^{\rm{TT}} (\frac{\rho}{(1 - \rho)^2} \frac{(\sigma^2_{\rm{a}} + \sigma^2_{\rm{c}})}{2}) \frac{\delta e}{c r}.
\end{equation}
Therefore, the price of a CR is the energy to charge $\delta e$ multiplied by the unit price $p(\rho)$, which is
\begin{equation}\label{unit_price}
  p(\rho) = \beta^{\rm{TT}} (\frac{\rho}{(1 - \rho)^2} \frac{(\sigma^2_{\rm{a}} + \sigma^2_{\rm{c}})}{2}) \frac{1}{c r}.
\end{equation}

Note that the unit price of a CR $p(\rho)$ is determined by the statistics of the CRs, and is dynamically and automatically updated as CRs are bought, without iterative processes as in \cite{Cui-2021, Wei-2017}. Moreover, the order of buying CRs can lead to different prices in the same timeframe, and the earlier the CR is bought, the lower the price is. Thus, EVs are encouraged to schedule their trip plans in advance and participate in the CRM. Nevertheless, this design makes deliberate arbitrage possible: since the price of CR will only rise but not fall, it is always profitable to buy a CR when the price is low and sell it later. Thus, users with no charging demand can also purchase CRs in the expected peak hours, which is not wanted. Although the detailed analysis of CR arbitrage is beyond the scope of this paper, we suppose that simple restrictions could be effective, such as limiting the times that a user can sell CRs within a month. In addition, since CRs can only be sold to other users and cannot be sold back to the CRP, arbitragers risk failing to sell their CRs.

The overall framework of the CRM is shown in Fig.~\ref{flowchart}. Note that the statistics of CRs may remain almost unchanged most of the time. Therefore, when updating the price and waiting time, $
\tau, \sigma^2_{\rm{a}}$, and $\sigma^2_{\rm{c}}$ can be updated only at regular intervals (e.g. every month) with historical data to reduce computational complexity.

\begin{figure}[!t]
  \centering
    \includegraphics[width=3.0in]{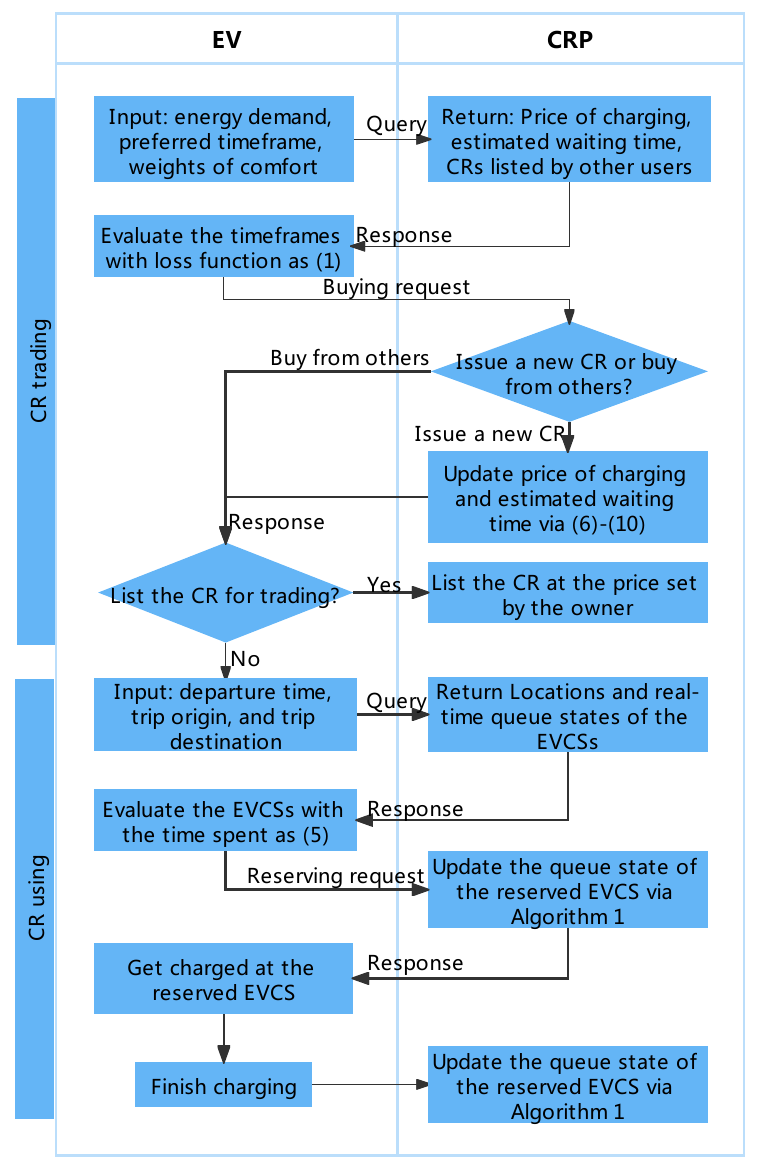}
  \caption{Flow Diagram of the Proposed CRM.}
  \label{flowchart}
\end{figure}

\section{Case Study}\label{sec_case}
In this section, the accuracy of the proposed methodology for waiting time estimation is verified, and the effect of the CRM on time cost reduction is investigated through simulations. The simulations have been carried out in the Opportunistic Network Environment \cite{keranen2008opportunistic}, and the basic scenario is built in the traffic network with an area of $4500 \times 3400 \ {\rm m^2}$ abstracted from the map of Helsinki city, where EVCSs are randomly deployed (Fig.~\ref{chap4-f0} and \ref{chap4-f01}), each provided with $6$ charging piles operating at $60$ kW.

EVs are randomly generated according to an exponential distribution
with mean $E_j$ in timeframe $j$, for $j = 1,...,12$. The durations of
timeframes are set at $\Delta t = 120$ min, and each simulation represents a
24-hour duration with 1 s resolution. The origin and destination $l^{\rm{o}},
l^{\rm{d}}$ of each EV are randomly selected on the map, while the energy to
charge $e$, moving speed on each street, and weights of comfort
$(\beta^{\rm{TT}}, \beta^{\rm{SDE}}, \beta^{\rm{SDL}})$ are independently generated
according to the distributions shown in Table~\ref{tab-EV-set} (${\rm
  U}$ for the uniform distribution). Simulations corresponding to each
parameter setting are repeated 1000 times.

\begin{figure}[!t]
  \centering
    \includegraphics[width=3.0in]{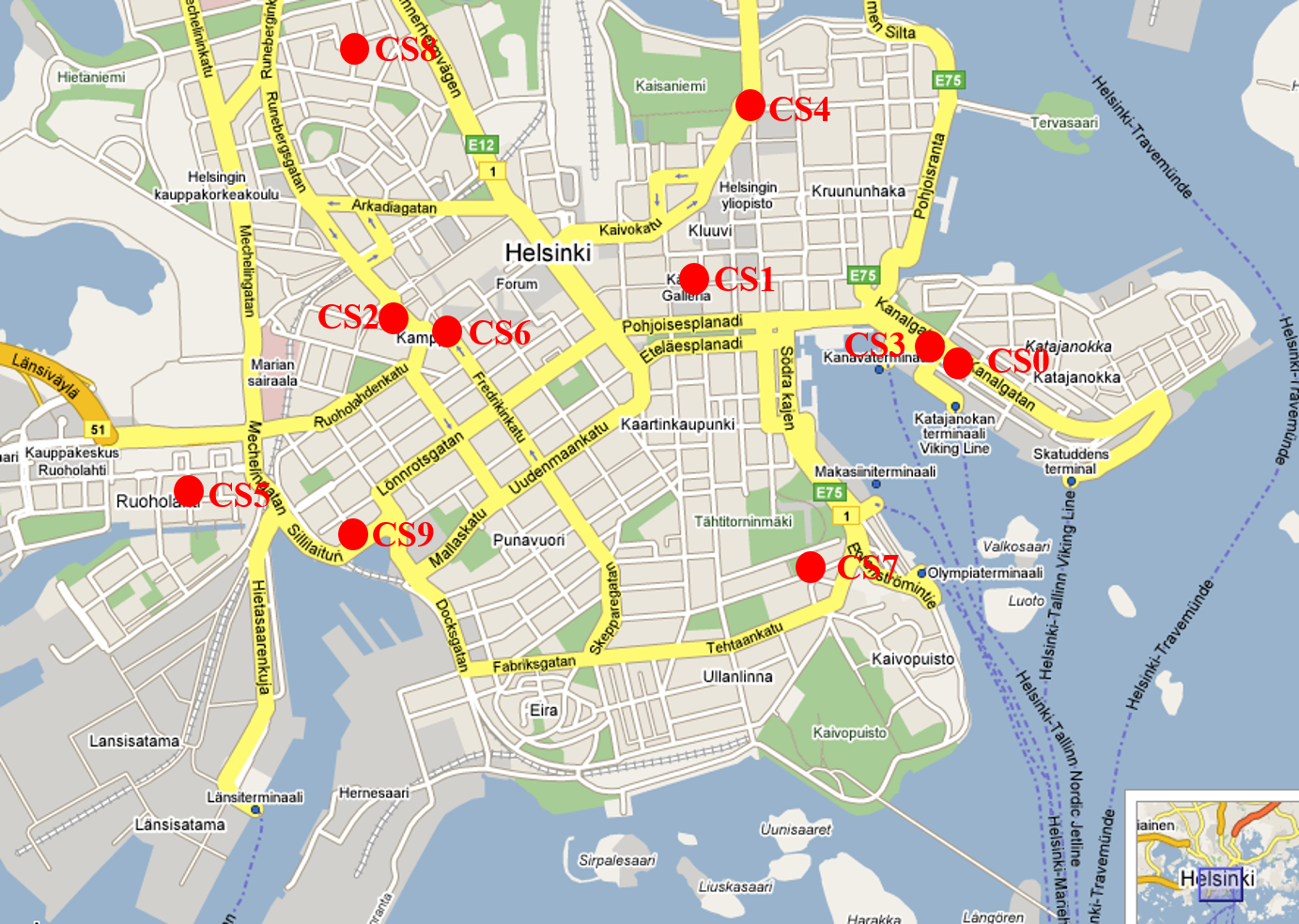}
  \caption{Map of Helsinki city and distribution of EVCSs.}
  \label{chap4-f0}
\end{figure}

\begin{figure}[!t]
  \centering
    \includegraphics[width=3.0in]{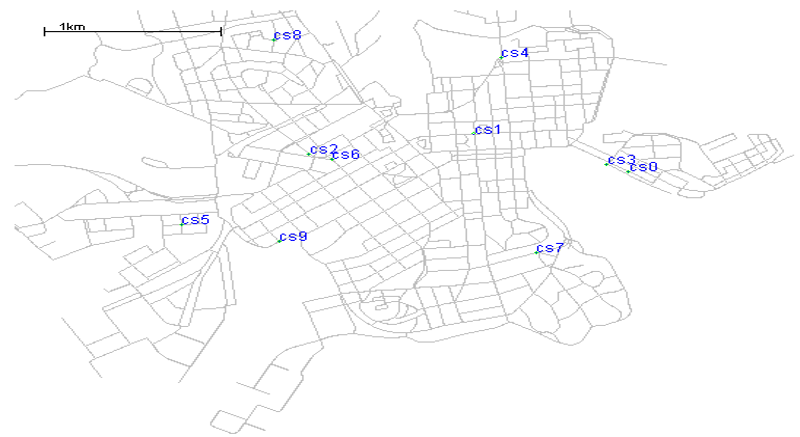}
  \caption{Simulation scenario of the urban PFC system.}
  \label{chap4-f01}
\end{figure}

\begin{table}[!t]
  \renewcommand{\arraystretch}{1.3}
  \caption{Distributions of EV Parameters}
  \label{tab-EV-set}
  \centering
  \begin{tabular}{c  c c}
    \toprule
    Item      & Distribution & Unit   \\
    \midrule
    Moving speed on each street & ${\rm U}(30, 50)$ & km/h \\
    Energy demand $e$   & ${\rm U}(10, 50)$ & kWh \\
    Coefficient for time $\beta^{\rm{TT}}$   & ${\rm U}(5, 15)$ & \$/h \\
    Coefficient for SDE $\beta^{\rm{SDE}}$   & ${\rm U}(0, 5)$ & \$/h \\
    Coefficient for SDL $\beta^{\rm{SDL}}$   & ${\rm U}(0, 5)$ & \$/h \\
    \bottomrule
  \end{tabular}
\end{table}

\subsection{Steady state}

\begin{figure}[!t]
  \centering
    \includegraphics[width=3.0in]{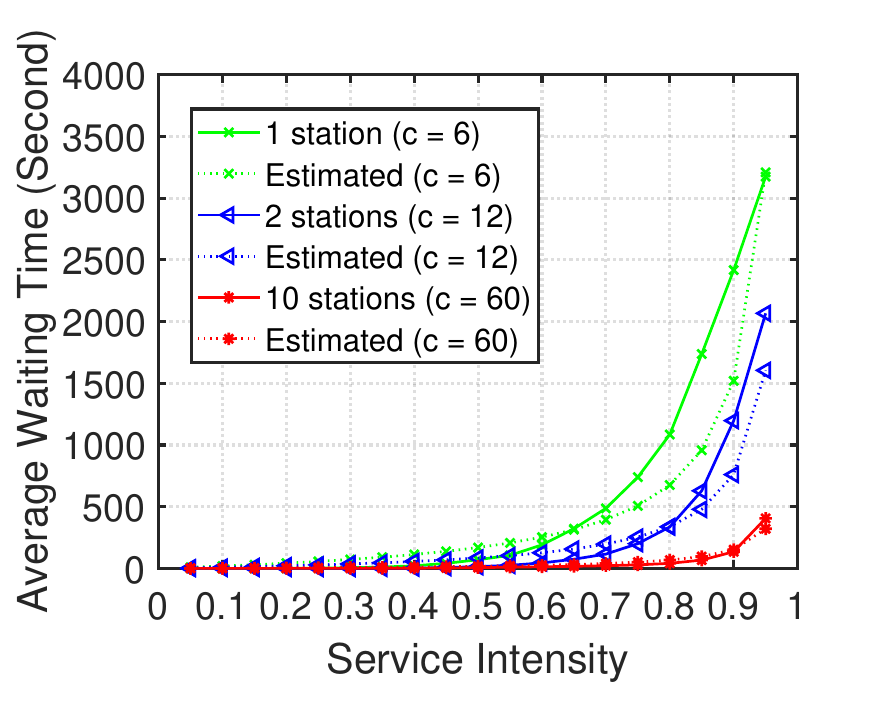}
  \caption{Average waiting time with respect to service intensity.}
  \label{chap4-f1}
\end{figure}

\begin{figure}[!t]
  \centering
    \includegraphics[width=3.0in]{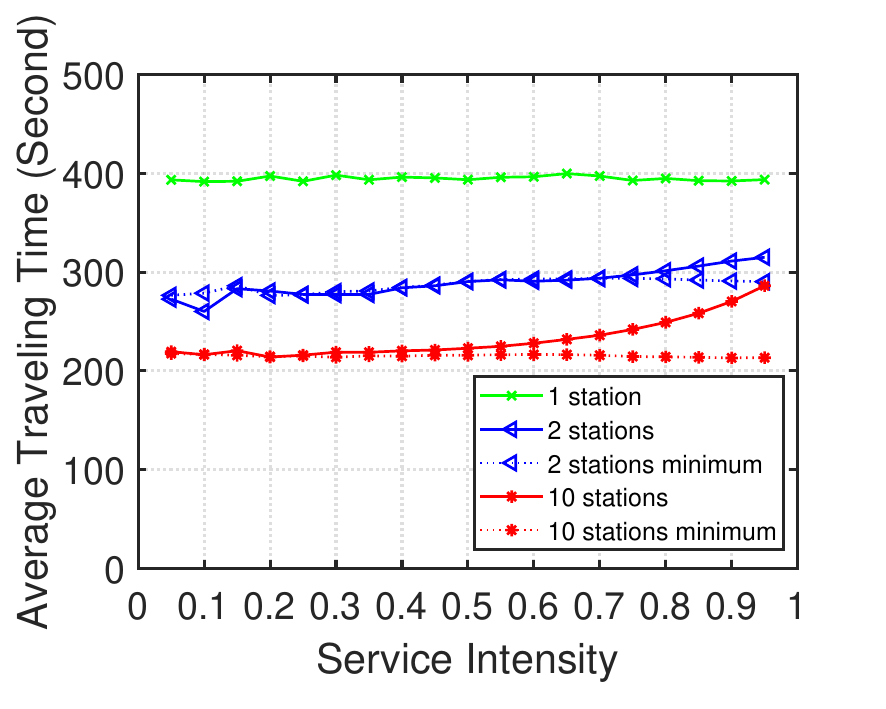}
    \caption{Average traveling rime with respect to service intensity.}
    \label{chap4-f2}
\end{figure}

First, in each simulation, set $E_j$ identically as $E$, for $j =
1,...,12$, to evaluate the charging system in steady state. Once an EV
is generated, it selects the EVCS with the minimum time cost (\ref{select-cs}), makes a reservation, and departs. Simulations are performed with 1, 2, and 10 stations randomly deployed, with a total number of charging piles $c$ of 6, 12, and 60, respectively. For comparison, service intensity $\rho = E\tau/cl$ is used to represent the setting of the simulations. Fig.~\ref{chap4-f1} shows the average waiting time in the charging system at different service intensities $\rho$, which increase drastically when $\rho$ approaches 1 (full capacity) in the three scenarios. Moreover, the average waiting time is approximately inversely proportional to the total number of charging piles $c$ in the charging system. For example, when $\rho = 0.95$, compared to the 1-station scenario (3174.3s), the average waiting time in the 2-station scenario (2066.3s) and the 10-station scenario (405.3s) is greatly reduced. In Fig.~\ref{chap4-f1}, the estimated waiting time (the dotted lines with the same color as the simulation results) is given by Equation (\ref{mean_waiting_time}), where $\tau, \sigma^2_{\rm{a}}$, and  $\sigma^2_{\rm{c}}$ are 30.01 min, 0.9773, and 0.1485, respectively, obtained from the simulation results through Equation (\ref{derive_lambda})-(\ref{derive_sigma}). Compared to the 10-station scenario, the estimation in the 2-station scenario is worse, as the difference in traveling time with respect to the EVCSs is larger. In Table~\ref{estimation-comparison}, the estimated waiting time by the Davidson function ($J = (\sigma^2_{\rm{a}} + \sigma^2_{\rm{c}})/2n = 0.0938$), and the M/M/n queueing model ($n = 6$) are compared. Both the Davidson function and the M/M/n queueing model only apply in the 1-station scenario, in which Equation (\ref{mean_waiting_time}) gives the same results as the Davidson function, as mentioned above. In the 10-station scenario, only our method still applies. Nevertheless, it tends to overestimate (underestimate) the waiting time at low (high) service intensity.

Fig.~\ref{chap4-f2} shows the average traveling time of the EVs, where
the dotted lines represent the corresponding scenarios in which EVs
select the EVCSs with the minimum traveling time for
comparison. Obviously, the results are the same in the 1-station
scenario, but in both the 2-station and 10-station scenarios, the
average traveling time is longer at higher service intensity, because
it is then more likely for the EVs to detour for less waiting
time. However, the increase in average traveling time is much less
than the increase in average waiting time, which is consistent with
the assumptions in Section~\ref{WTE}. For example, in the 10-station
scenario, as the service intensity rises from 0.9 to 0.95, an increase
of 270.2 s is observed in the average waiting time, while the average
traveling time only increases by 15.9 s.

\begin{table}[!t]
  \renewcommand{\arraystretch}{0.9}
  \caption{Comparison of Estimated Waiting Time (s)}
  \label{estimation-comparison}
  \centering
  \begin{tabular}{c|c|c|c|c}
  \toprule
  Service intensity             & 0.80      & 0.85     & 0.90      & 0.95     \\ \midrule
  Simulation results (1-station)            & 1084.4 & 1736.8 & 2416.9 & 3174.3 \\ \midrule
  Davidson function (1-station)& 675.5   & 956.9  & 1519.8 & 3208.5 \\ \midrule
  M/M/n model (1-station)     & 776.7 & 1248.1 & 2220.4 & 5193.5 \\ \midrule
  Simulation results (10-station)      & 41.1 & 67.1 & 135.1 & 405.3 \\ \midrule
  Equation (\ref{mean_waiting_time})(10-station)           & \textbf{67.5 }   & \textbf{95.7}  & \textbf{152.0} & \textbf{320.9} \\ \bottomrule
  \end{tabular}
\end{table}

\subsection{Typical day}
\begin{table}[!t]
  \renewcommand{\arraystretch}{0.9}
  \caption{Typical Day Settings of the Timeframes}
  \label{tab-time-frame-set}
  \centering
  \begin{tabular}{c  c c}
    \toprule
    Timeframe      & Mean Number of EVs & Electricity Price (\$/kWh)   \\
    \midrule
    1:00$-$3:00   & 24 & 0.24 \\
    3:00$-$5:00   & 24 & 0.24 \\
    5:00$-$7:00   & 120 & 0.24 \\
    7:00$-$9:00   & 192 & 0.24 \\
    9:00$-$11:00   & 264 & 0.30 \\
    11:00$-$13:00   & 192 & 0.36 \\
    13:00$-$15:00   & 192 & 0.36 \\
    15:00$-$17:00   & 192 & 0.30 \\
    17:00$-$19:00   & 264 & 0.36 \\
    19:00$-$21:00   & 192 & 0.36 \\
    21:00$-$23:00   & 120 & 0.30 \\
    23:00$-$1:00   & 24 & 0.24 \\
    \bottomrule
  \end{tabular}
\end{table}

\begin{figure*}[!t]
  \centering
  \subfloat[]{
     \includegraphics[width=0.24\linewidth]{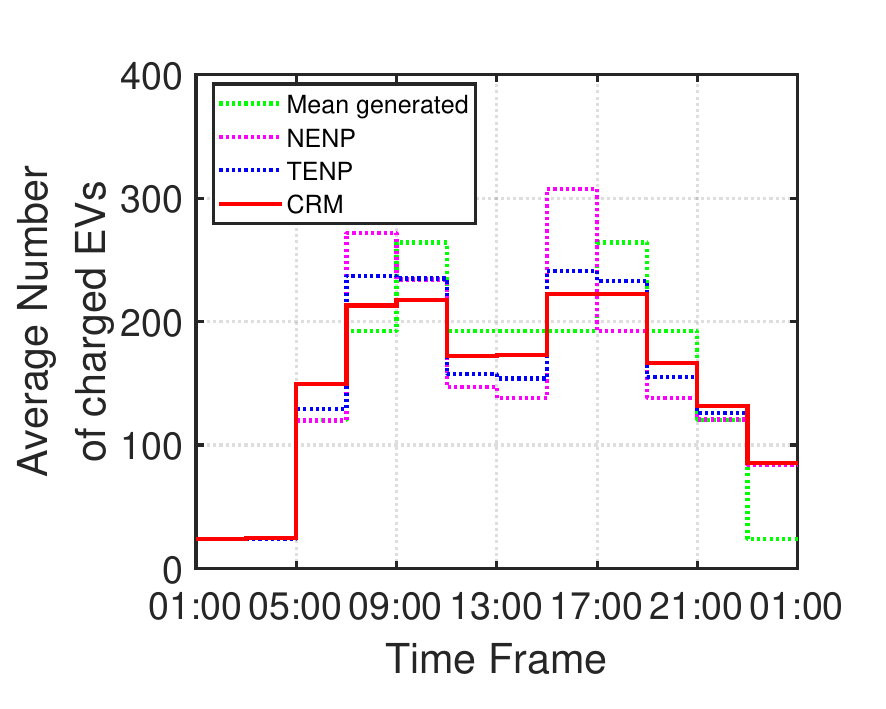}}
  \hfill
  \subfloat[]{
      \includegraphics[width=0.24\linewidth]{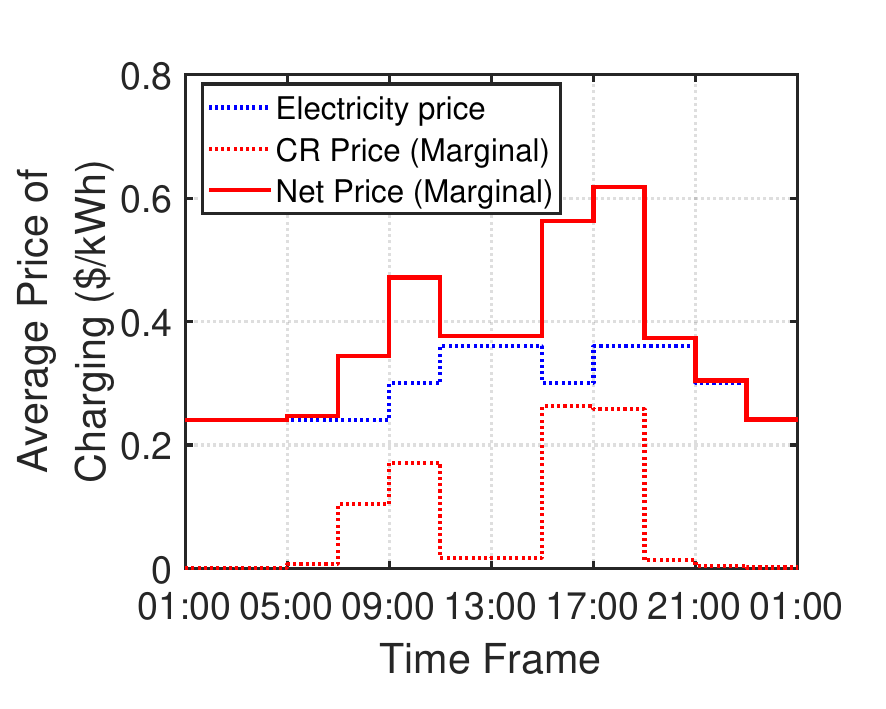}}
  \subfloat[]{
      \includegraphics[width=0.24\linewidth]{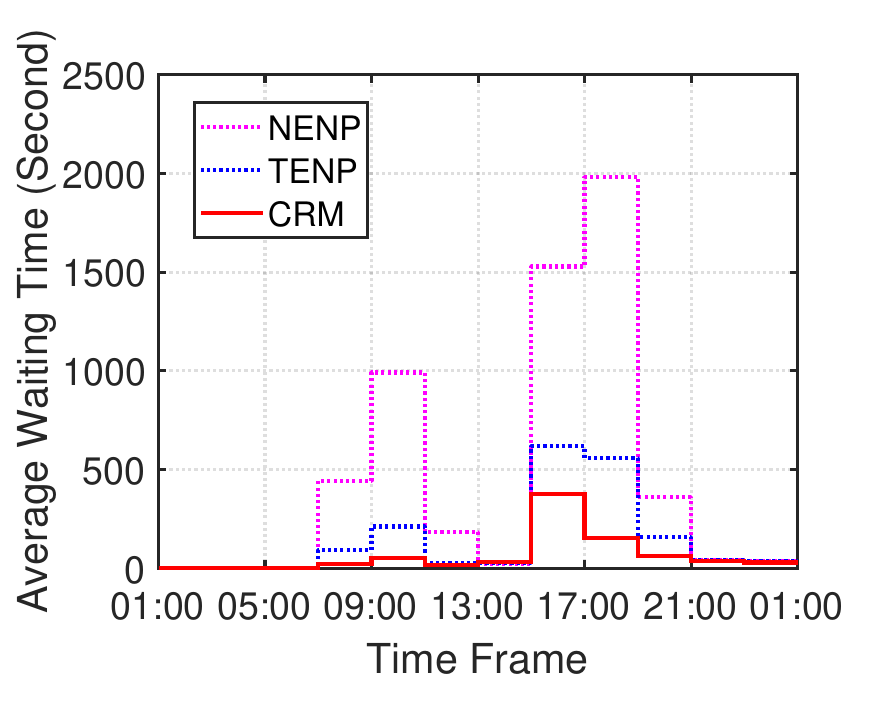}}
  \hfill
  \subfloat[]{
      \includegraphics[width=0.24\linewidth]{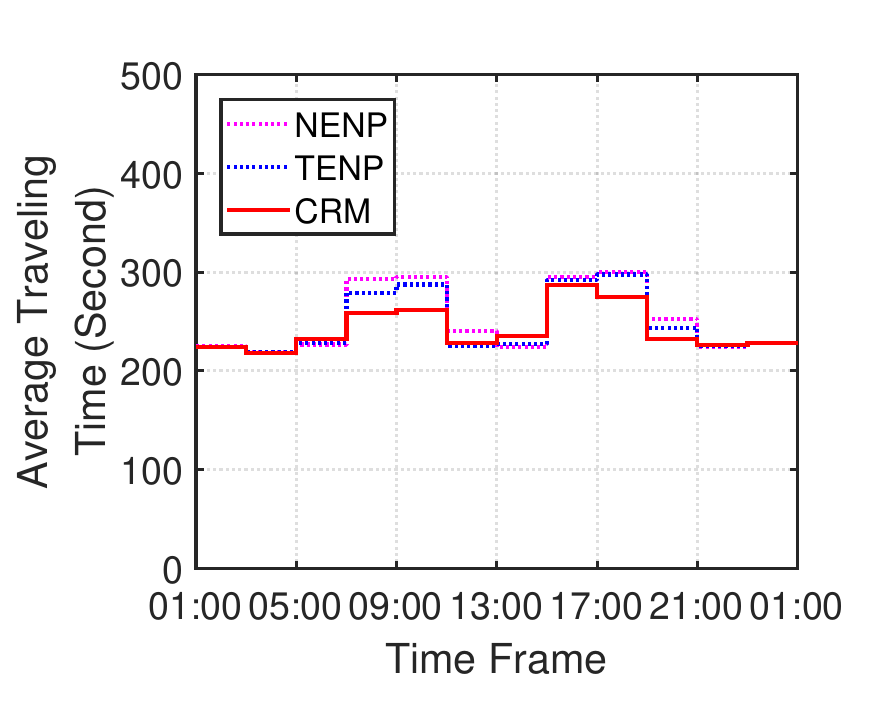}}
  \caption{Typical day: (a) Average number of charged EVs in the timeframes. (b) Average marginal price of charging in the timeframes. (c) Average waiting time of the timeframes. (d) Average traveling time of the timeframes.}
  \label{chap4-f5} 
\end{figure*}

\begin{table*}[!t]
  \renewcommand{\arraystretch}{1.0}
  \caption{Total Time Spent and Cost Induced on Typical Day \\ (Mean Number of EVs: 1800)}
  \label{tab-total-cost}
  \centering
  \resizebox{\textwidth}{11mm}{
    \begin{tabular}{c|ccc|ccccc|ccc}
      \hline
     & \multicolumn{3}{c|}{Time spent (h)} & \multicolumn{5}{c|}{Costs induced (\$)} & \multicolumn{3}{c}{Overall (\$)} \\ \hline
    Scheme & \multicolumn{1}{c|}{\begin{tabular}[c]{@{}c@{}}Charging \\ time\end{tabular}} & \multicolumn{1}{c|}{\begin{tabular}[c]{@{}c@{}}Waiting\\  time\end{tabular}} & \begin{tabular}[c]{@{}c@{}}Traveling\\ time\end{tabular} & \multicolumn{1}{c|}{\begin{tabular}[c]{@{}c@{}}Charging\\ time cost\end{tabular}} & \multicolumn{1}{c|}{\begin{tabular}[c]{@{}c@{}}Waiting\\ time cost\end{tabular}} & \multicolumn{1}{c|}{\begin{tabular}[c]{@{}c@{}}Traveling\\ time cost\end{tabular}} & \multicolumn{1}{c|}{\begin{tabular}[c]{@{}c@{}}Electricity\\ cost\end{tabular}} & \begin{tabular}[c]{@{}c@{}}Scheduled\\ delay cost\end{tabular} & \multicolumn{1}{c|}{\begin{tabular}[c]{@{}c@{}}Total\\ cost\end{tabular}} & \multicolumn{1}{c|}{\begin{tabular}[c]{@{}c@{}}Charging\\ right payment\end{tabular}} & \begin{tabular}[c]{@{}c@{}}Cost\\ reduction\end{tabular} \\ \hline
    NENP & \multicolumn{1}{c|}{899.6} & \multicolumn{1}{c|}{163.9} & 132.2 & \multicolumn{1}{c|}{9012.1} & \multicolumn{1}{c|}{1643.5} & \multicolumn{1}{c|}{1323.7} & \multicolumn{1}{c|}{16595.5} & 193.0 & \multicolumn{1}{c|}{28767.8} & \multicolumn{1}{c|}{0.0} & 0.0 \\ \hline
    TENP & \multicolumn{1}{c|}{900.5} & \multicolumn{1}{c|}{95.0} & 130.6 & \multicolumn{1}{c|}{9020.7} & \multicolumn{1}{c|}{954.1} & \multicolumn{1}{c|}{1308.3} & \multicolumn{1}{c|}{16680.1} & 172.5 & \multicolumn{1}{c|}{28135.7} & \multicolumn{1}{c|}{0.0} & 632.1 \\ \hline
    CRM & \multicolumn{1}{c|}{899.8} & \multicolumn{1}{c|}{49.7} & 126.6 & \multicolumn{1}{c|}{9012.6} & \multicolumn{1}{c|}{501.7} & \multicolumn{1}{c|}{1268.8} & \multicolumn{1}{c|}{16732.0} & 398.3 & \multicolumn{1}{c|}{27913.3} & \multicolumn{1}{c|}{521.7} & 854.5 \\ \hline 
    \end{tabular}
    }
\end{table*}

Second, we inspect the effect of the CRM on time cost reduction on a typical day, where the mean generated number of EVs and the static electricity price for charging in the timeframes are set to be different, as shown in Table~\ref{tab-time-frame-set}. Once an EV is generated, it first selects the timeframe to get charged (\ref{select-time-frame}), where the preferred timeframe is set to that of its generation), then departs randomly in the selected timeframe and selects an EVCS to get charged (\ref{select-cs}). For comparison, the following schemes are evaluated:

\paragraph{NENP} No waiting time estimation is available in advance, and the price of charging is merely set at the electricity price. 

\paragraph{TENP} Using the information from the CRs, waiting time
estimation is available in advance, but the charging price is still set at the electricity price.

\paragraph{CRM} (The proposed mechanism) Waiting time estimation is
available in advance, and the price of a CR (\ref{sale_price}) is integrated into the price of charging. For the sake of discussion, CRs are used by owners instead of being sold to others.

Figs.\ref{chap4-f5}(a) and \ref{chap4-f5}(b) present the average number of
charged EVs and the price of charging with respect to the
timeframes. Figs.\ref{chap4-f5}(c) and \ref{chap4-f5}(d) present the average
waiting time and average traveling time of the timeframes. In the NENP
scheme (pink dotted lines), the number of charged EVs deviates from
the original generation (green line in Fig.~\ref{chap4-f5}(a)) and has
reached 272 and 307 in 7:00$-$9:00 and 15:00$-$17:00, respectively,
due to the relatively low electricity price, which results in an
average waiting time of 992.1 s and 1528.2 s, respectively. With the
help of waiting time estimation in the TENP scheme (blue dotted
lines), such deviation, as well as the number of EVs in the peak hours
in 17:00$-$19:00, is substantially reduced, resulting in waiting time
and traveling time reduction for the timeframes. The waiting time and
the traveling time in the peak hours is further reduced by the price
mechanism of the CRM (red lines), where the EVs are shifted to those
timeframes with fewer purchased CRs. The maximum average waiting time
(378.2 s) in the CRM appears in 15:00$-$17:00, with the maximum CR price of \$ 0.26 /kWh in the same timeframe. This is due to the relatively low electricity price in the timeframe and the shifted EVs from other timeframes, but the average waiting time is reduced by 75.3\% and 39.0\% compared to the NENP scheme and the TENP scheme, respectively. 

The total time spent and costs induced in the typical-day scenario are listed in Table~\ref{tab-total-cost}. The total charging time is approximately 900 h in the three schemes, as the average charging time is 0.5 h. As the average coefficient for time cost is set as \$10/h, the induced time costs are approximately the corresponding time spent multiplying 10. Compared to the NENP scheme, the waiting time cost is reduced by 42.0\% and 69.7\% in the TENP and CRM, respectively. The reduction in traveling time is less significant at 1.2\% and 4.2\% in the TENP and CRM, respectively. Sacrifices include the increased SD cost, as more EVs deviate from their preferred timeframe for charging due to the high prices, and the increased electricity fees. The latter is because the price of CR makes some low-electricity-price timeframes less attractive. Nevertheless, the total cost in the CRM is reduced by \$ 854.5, compared to the NENP scheme. Although the reduction seems small compared to the total electricity fees and the charging time cost, it accounts for 52.0\% of the waiting time cost (\$ 1643.5). Intuitively, the time costs can be further reduced considering CR trading among EV users, because users who are more time-sensitive can purchase CRs from others instead of paying the high CR price in peak hours or deviating from their preferred timeframes for charging.

In the CRM, the accumulated CR payment (\$ 521.7) is close to the waiting time cost (\$ 501.7), which indicates that the price of CRs can well reflect the waiting time cost. In addition, the CR payment is much smaller than the total electricity fees (\$ 16732.0), although the marginal CR price can approach the electricity price in certain timeframes (as shown in Fig.~\ref{chap4-f5}(b)). In Table~\ref{tab-total-cost}, CR payment is not considered in the total cost, because, as mentioned above, the money can be used to compensate for electricity fees or to motivate EVs to avoid hours with peak charging demand, although the detailed method is beyond the scope of this paper.

\section{Conclusion}\label{sec_conclusion}
In this article, the concepts of EV CRs and a CRM are proposed. By purchasing CRs and participating in the CRM, EVs can reduce their charging waiting time and better plan their trip. The external-cost-based CR pricing can guide EVs toward optimized charging behaviors, greatly reducing the time cost in the charging system. Real-time EVCS queue states are utilized in the CRM, which can meaningfully reduce the average waiting time, and the impact is considered in our waiting time modeling. Instead of relying on forecasts, the statistics of the CRs are used to provide waiting time estimation and charging pricing in advance. The case study verified the proposed framework for waiting time estimation and the effect of the CRM on reducing time cost in EV charging.

This article focuses on reducing time cost in PFC charging systems,
and the electricity price is assumed to be static. In the future, the proposed dynamic and autonomous pricing mechanism can be extended to consider the interdependence of transportation networks and PDNs. Furthermore, the premise in Corollary \ref{corollary_equal} may not hold if there are distant EVCSs. Although verified in a charging system where the EVCSs are up to several kilometers apart, our method may need modification for applying in larger charging systems. These will be further studied in future works.

\appendices
\section{Proof of Theorem\ref{theorem_equal}}\label{theorem_equal_app}
\begin{proof}
  The queue state of the whole charging system is denoted as ${\rm queue} = \{t^{\rm{f}}_i|i = 1, ... , I\}$, where $t^{\rm{f}}_i$ and $I$ are the time for EV $i$ to finish charging (according to the ascending order of $t^{\rm{f}}_i$) and the number of arrived EVs in the system, respectively. Let the initial queue state of the OS and the IS be equal: ${\rm queue(OS)} = {\rm queue(IS)}$. Their queue states only change when an EV arrives at the system or finishes charging:

  When EV $I + 1$ with $e_{I + 1}$ to charge arrives in the OS at
  $t^{\rm{a}}_{I + 1}$, the same EV $I + 1$ with $e_{I + 1}$ to charge also
  arrives in the IS at $t^{\rm{a}}_{I + 1}$, according to our definition of
  the IS. If the queue states of the OS and the IS are equal, there
  are two possible circumstances: (1) If $I < c$, there must be an
  available charging pile in both the OS and the IS, and therefore, we have $t^{\rm{f}}_{I + 1} = t^{\rm{a}}_{I + 1} + e_{I + 1}/r$ for both systems. For the OS and the IS, add $t^{\rm{f}}_{I + 1}$ to queue, and their queue states are still equal. (2) If $I \ge c$, all piles have been occupied in the two systems. In the IS, EV $I + 1$ has to wait for an available pile until EV $I - c +1$ finishes charging at $t^{\rm{avail}}(IS) = t^{\rm{f}}_{I - c +1}$. In the OS, EV $I + 1$ arrives at EVCS $k^*$ with the earliest available time, i.e., $k^* = {\rm argmin} \ t^{\rm{avail}}(k)$. Note that the earliest available time among the EVCSs is also the available time of the whole OS, which is $t^{\rm{avail}}(OS) = t^{\rm{f}}_{I - c +1}$. Therefore, for both systems, EV $I + 1$ waits for $t^{\rm{w}}_{I+1} = t^{\rm{f}}_{I - c +1} - t^{\rm{a}}_{I + 1}$ and finishes charging at $t^{\rm{f}}_{I + 1} = t^{\rm{f}}_{I - c +1} + e_{I + 1}/r$. Add $t^{\rm{f}}_{I + 1}$ to queue, and the queue states of the OS and the IS are still equal.

  If the queue states of the OS and the IS are equal, we have
  $t^{\rm{f}}_i(OS) = t^{\rm{f}}_i(IS), i = 1,...,I$. Therefore, when EV $i$
  finishes charging in the OS, EV $i$ in the IS also finishes charging. Remove $t^{\rm{f}}_i$ from queue, and their queue states remain equal.
   
  Thus, the queue states of the OS and IS are always equal since the
  initial time, and the waiting times of the EVs in the two systems are
  also equal. Therefore, the average waiting time of the OS, which is
  independent of the initial state, is equal to that of the IS.
\end{proof}

\section{Proof of Corollary\ref{corollary_equal}}\label{corollary_equal_app}
\begin{proof}
  In the proof of Theorem \ref{theorem_equal}, consider the second
  circumstance ($I \ge c$) when EV $I + 1$ with $e_{I + 1}$ to charge
  arrives in the two systems at $t^{\rm{a}}_{I + 1}$. In the OS, EV $I + 1$
  arrives at EVCS $k^*$ with the minimum total time, i.e., $k^* = {\rm
    argmin} \ t^{\rm{w}}(k) + t^{\rm{oc}}(k) + t^{\rm{cd}}(k)$. Since the difference in
  traveling time $t^{\rm{oc}}(k) + t^{\rm{cd}}(k)$ of different EVCSs is
  negligible compared to the difference in waiting time $t^{\rm{w}}(k)$, we
  have $k^* = {\rm argmin} \ t^{\rm{w}}(k)$. Note that $t^{\rm{w}}(k) = t^{\rm{avail}}(k)
  - t^{\rm{a}}_{I + 1}$ and $t^{\rm{a}}_{I + 1}$ is given; thus, $k^* = {\rm argmin}
  \ t^{\rm{avail}}(k)$, which is the same as in Theorem
  \ref{theorem_equal}. Therefore, follow the proof in
  Appendix.\ref{theorem_equal_app}, and the average waiting time of
  the OS is equal to that of the IS.
\end{proof}

\bibliographystyle{IEEEtran}
\bibliography{reference}

\end{document}